\documentclass{article}
\usepackage{amsmath, amssymb, amsthm}
\usepackage{setspace}
\usepackage[a4paper, margin = 2.54cm]{geometry}

\usepackage[final,pdfusetitle, pdftex]{hyperref}
\hypersetup{pdfauthor={Gonzalo Bley and S{\o}ren Fournais}
}

\usepackage{authblk}
\usepackage{graphicx}
\usepackage{fancyhdr}
\usepackage{color}

\newtheorem{theorem}{Theorem}
\newtheorem{lemma}[theorem]{Lemma}
\newtheorem{proposition}[theorem]{Proposition}

\newtheorem{remark}[theorem]{Remark}

\numberwithin{equation}{section}
\numberwithin{theorem}{section}
\DeclareMathOperator{\supp}{\mathrm{supp}}
\DeclareMathOperator{\tr}{\mathrm{tr}}

\begin{document}
\title{The Magnetic Scott Correction for\\Relativistic Matter at Criticality}
\author[1]{Gonzalo A. Bley \thanks{gonzalo.bley@gmail.com}}
\author[1]{S{\o}ren Fournais \thanks{fournais@math.au.dk}}
\affil[1]{Department of Mathematics, Aarhus University, \authorcr Ny Munkegade 118,  8000 Aarhus C, Denmark}
\maketitle
\begin{abstract}
We provide a proof of the first correction to the leading asymptotics of the minimal energy of pseudo-relativistic molecules in the presence of magnetic fields, the so-called ``relativistic Scott correction,'' when $\max{Z_k\alpha} \leq 2/\pi$, where $Z_k$ is the charge of the $k$-th nucleus and $\alpha$ is the fine structure constant. Our theorem extends a previous result by Erd\H{o}s, Fournais, and Solovej to the critical constant $2/\pi$ in the relativistic Hardy inequality $|p| - \frac{2}{\pi |x|} \geq 0$.
\end{abstract}
\begin{section}{Introduction}
\begin{subsection}{Primer on Thomas-Fermi theory}
In this subsection we shall provide a brief introduction to Thomas-Fermi theory that will be sufficient for our purposes in this paper. The interested reader is referred to more extended treatments, such as \cite{LS} or \cite{L}, for more information on the subject.

We will consider a neutral molecule composed of $M$ nuclei at positions $(R_1, R_2, \ldots, R_M) \in \mathbb{R}^3$ and charges $(Z_1, \ldots, Z_M) \in \mathbb{R}_+^3$, and $N$ non-relativistic electrons at positions $(x_1, x_2, \ldots, x_N) \in \mathbb{R}^3$. The quantum-mechanical Hamiltonian of the system will then be
\begin{gather}
H \equiv \sum_{n = 1}^N\frac{p_n^2}{2} - \sum_{n = 1}^N V(x_n) + \sum_{n < m}|x_n - x_m|^{-1},
\end{gather}
where $p_n = -i\nabla_n$ and $V(x) = \sum_{m = 1}^M Z_m/|x - R_m|$, acting on $\bigwedge_{i = 1}^N L^2(\mathbb{R}^3)\otimes \mathbb{C}^2$. ($\bigwedge$ denotes antisymmetric tensor product.) The ions are assumed to be fixed in space (Born-Oppenheimer approximation) and the interionic repulsion term $\sum_{n < m}Z_nZ_m/|x_n - x_m|$ is omitted. Units are such that $m = \hbar = e = 1$, $m$ being the electronic mass and $e$ the electronic charge. Spin for the electrons will be considered, with 2 possible states, which is the reason for the term $\mathbb{C}^2$ in the antisymmetric tensor product above. Note that $N = Z_1 + \ldots + Z_M \equiv Z$ by neutrality.

After making certain assumptions and simplifications (see \cite[Section I.1]{LS} for more details), one can approximate the ground-state energy of the system using a functional of the electronic density $\rho$, defined as
\begin{gather}
\mathcal{E}(\rho) \equiv \frac{3}{10}\left(3\pi^2\right)^{2/3}\int_{\mathbb{R}^3}\rho(x)^{5/3}\,dx - \int_{\mathbb{R}^3}V(x)\rho(x)\,dx + \frac{1}{2}\int_{\mathbb{R}^3}\!\int_{\mathbb{R}^3}\frac{\rho(x)\rho(y)}{|x - y|}\,dx\,dy.
\label{equation.thomas.fermi.functional}
\end{gather}
The functional
$\mathcal{E}$ is known as the Thomas-Fermi functional after L.H. Thomas \cite{T} and E. Fermi \cite{F}. The constant in front of $\int\rho^{5/3}$ is, in general, $3(6\pi^2/q)^{2/3}/10$, with $q$ equal to the number of spin states. The domain of definition of $\mathcal{E}$ is $\rho \in L^{5/3}(\mathbb{R}^3)\cap\,L^1(\mathbb{R}^3)$ with $\rho \geq 0$. By definition of the density, we will have $\int \rho = Z$. Under the assumptions explained in \cite[Section I.1]{LS},
\begin{gather}
E^{TF} \equiv \inf\left\{\mathcal{E}(\rho) : \rho \geq 0, \rho \in L^{5/3}(\mathbb{R}^3)\cap\,L^1(\mathbb{R}^3), \int_{\mathbb{R}^3}\rho(x)\,dx = Z\right\},
\end{gather}
should be a good approximation to the ground-state energy of the original Hamiltonian $H$. Before discussing if that is the case or not, let us address $E^{TF}$ in more detail. The most important facts for us at the moment are that $-\infty < E^{TF} < 0$; the infimum is actually a minimum; this minimum is unique. We shall call the minimum the Thomas-Fermi density, $\rho^{TF}$. The proof of these results can be found in $\cite{LS}$.

If one fixes $\lambda > 0$ and performs the scaling
\begin{align}
V_{\lambda} & \equiv \lambda^{4/3}V(\lambda^{1/3}x),\\
\rho_{\lambda} & \equiv \lambda^2\rho(\lambda^{1/3}x),
\end{align}
then it is easy to verify that $\mathcal{E}\left(V_{\lambda}; \rho_{\lambda}\right) = \lambda^{7/3}\mathcal{E}\left(V; \rho\right)$. If we define the vectors
\begin{align}
z & \equiv Z^{-1}(Z_1, \ldots , Z_M),
\label{definition.z}\\
r & \equiv Z^{1/3}(R_1, \ldots, R_M),
\label{definition.r}
\end{align}
and the functions $V_{r, z}$ and $\rho_z$ through the relations
\begin{align}
V(x) & \equiv Z^{4/3}V_{r, z}(Z^{1/3}x),\\
\rho(x) & \equiv Z^2\rho_z(Z^{1/3}x),
\end{align}
which implies $\mathcal{E}(V; \rho) = Z^{7/3}\mathcal{E}(V_{r, z}, \rho_z)$, we then have
\begin{gather}
E^{TF} = Z^{7/3}E^{TF}(z, r).
\label{equation.thomas.fermi.scaling}
\end{gather}
It turns out that $E^{TF}$ yields exactly the quantum-mechanical ground-state energy $E \equiv \text{inf spec }H$ in the limit $Z \to \infty$. The precise statement of the result is as follows.
\begin{theorem}[Lieb, Simon, 1977 \cite{LS}]
\label{theorem.lieb.simon}
As $a \to \infty$ so that $a(Z_1 + \ldots + Z_M)$ is an integer,
\begin{gather}
E(aZ_1, \ldots, aZ_M; a^{-1/3}R_1, \ldots, a^{-1/3}R_M) = a^{7/3}E^{TF}(Z_1, \ldots, Z_M; R_1, \ldots, R_M) + o(a^{7/3}).
\end{gather}
\end{theorem}
\begin{remark}
Theorem \ref{theorem.lieb.simon} corresponds to \cite[Theorem III.1]{LS}. The version appearing in \cite{LS} is more general.
\end{remark}
We now define the Thomas-Fermi potential as
\begin{align}\label{eq:VTF}
V^{TF}(x) \equiv V(x) - \int_{\mathbb{R}^3}\frac{\rho^{TF}(y)}{|x - y|}\,dy,
\end{align}
which satisfies, with $d_R(x) \equiv \min_{1 \leq k \leq M}{|x - R_k|}$ and $R_{\min} \equiv \min_{i \neq j}|R_i - R_j|$,
\begin{align}
\left|V^{TF}(x) - \frac{Z_k}{|x - R_k|}\right| & \leq C\left(Z^{-1/3}R_{\min}^{-1} + 1\right)Z^{4/3}\qquad\text{when $|x - R_k| \leq R_{\min}/2$,}
\label{estimate.VTF.Coulomb}\\
\left|V^{TF}(x)\right| & \leq \frac{CZ}{d_R(x)},
\label{estimate.VTF.asymptotic.Coulomb}\\
\left|V^{TF}(x)\right| & \leq \frac{C}{d_R(x)^{4}}.
\label{estimate.VTF.asymptotic.Coulomb.4th.power}
\end{align}
The constant $C$ in \eqref{estimate.VTF.Coulomb}, \eqref{estimate.VTF.asymptotic.Coulomb} and \eqref{estimate.VTF.asymptotic.Coulomb.4th.power} is universal (a number). For more information on these estimates on $V^{TF}$ the reader is referred to \cite[Subsection 2.2]{SSS}, \cite[Section 2]{SS}, and \cite[Section IV]{LS}. $V^{TF}$ appears naturally in the Euler-Lagrange equation for \eqref{equation.thomas.fermi.functional}, namely
\begin{gather}
\frac{1}{2}(3\pi^2)^{2/3}\left(\rho^{TF}\right)^{2/3} = V(x) - \int_{\mathbb{R}^3}\frac{\rho^{TF}(y)}{|x - y|}\,dy.
\end{gather}
(See \cite[Equation (14)]{SS}.) Another property of $V^{TF}$ that we shall use is that the semiclassical approximation to the sum of the negative eigenvalues of $p^2/2 - V^{TF}$, acting on $L^2(\mathbb{R}^3)\otimes \mathbb{C}^2$, is given by
\begin{gather}
\frac{2}{(2\pi)^3}\int_{\mathbb{R}^3}\!\int_{\mathbb{R}^3}\left(\frac{p^2}{2} - V^{TF}\right)_-\,dx\,dp = -\frac{4\sqrt{2}}{15\pi^2}\int_{\mathbb{R}^3}V^{TF}(x)^{5/2}\,dx = E^{TF} + D(\rho^{TF}),
\label{equation.semiclassical.V.TF}
\end{gather}
(see \cite[Equation (27)]{SS}).

We finish this subsection by recording two scaling properties of $V^{TF}$ and $\rho^{TF}$, which follow from our discussions before Theorem \ref{theorem.lieb.simon}: for $a > 0$,
\begin{align} \label{eq:TF-scaling}
V^{TF}_{Z, R}(x) & = a^4 V_{a^{-3}Z, aR}^{TF}(ax),\\
\rho_{Z, R}^{TF}(x) & = a^6 \rho_{a^{-3}Z, aR}^{TF}(ax).
\end{align}
Here $Z$ and $R$ take on the meaning $(Z_1, \ldots, Z_M)$ and $(R_1, \ldots, R_M)$, respectively. This is a notation we shall use ocassionally (when there is no risk of confusion).
\end{subsection}
\begin{subsection}{The Scott correction}
Using Theorem \ref{theorem.lieb.simon} while setting $M = 1$, the atomic case, yields
\begin{gather}
E = -C_{TF}Z^{7/3} + o(Z^{7/3})
\end{gather}
for some positive $C_{TF}$. In 1952, Scott \cite{S} predicted a correction to the energy asymptotics of large atoms, $-C_{TF}Z^{7/3}$, proportional to $Z^2$. This was finally proved rigorously by Siedentop and Weikard \cite{SW1, SW2, SW3} (upper and lower bounds) and Hughes \cite{H} (lower bound). They found that, for a neutral atom, (remember that we work in a setting of $2$ spin states)
\begin{gather}
E = -C_{TF}Z^{7/3} + \frac{Z^2}{2} + o(Z^2).
\end{gather}
This second term, $Z^2/2$, is known as the Scott correction. Siedentop and Weikard's result was later extended by Ivrii and Sigal \cite{IS} to the case of neutral molecules and to the case of ions by Bach \cite{Bach}. The result for molecules reads as
\begin{gather}
\label{result.ivrii.sigal}
E(Z_1, \ldots, Z_M, R_1, \ldots, R_M) = E^{TF} + \frac{1}{2}\sum_{m = 1}^M Z_m^2 + o(Z^2),
\end{gather}
where $(Z_1, \ldots, Z_M)$ and $(R_1, \ldots, R_M)$ are the charges and positions of the $M$ nuclei of the molecule, respectively. Here, $E^{TF}$ is the infimum of the Thomas-Fermi functional we saw before, with $V(x)$ equal to $\sum_{m = 1}^M Z_m/|x - R_m|$. The Scott correction for neutral molecules was also proven later by Solovej and Spitzer \cite{SS} using a generalisation of coherent states. We remark that there are some extra technical assumptions regarding \eqref{result.ivrii.sigal}, which the reader may check in \cite[Theorem 0.1]{IS}.

We finally mention a result by Fefferman and Seco \cite{FS}, where they proved that for a neutral atom of nuclear charge $Z$,
\begin{gather}
E = -C_{TF}Z^{7/3} + \frac{Z^2}{2} - C_2 Z^{5/3} + o(Z^{5/3}),
\label{equation.fefferman.seco}
\end{gather}
finding in this way, in the atomic case, a second correction, beyond the Scott term, proportional to $Z^{5/3}$. This term had been predicted by Dirac \cite{D} and Schwinger \cite{Sch}. $C_2$ is here a positive constant.
\end{subsection}
\begin{subsection}{The relativistic Scott correction}
We now review the large-$Z$ asymptotics for pseudo-relativistic atoms and molecules. Pseudo-relativistic means that the kinetic energy operator $p^2/2$ for electrons is replaced by $\sqrt{p^2\alpha^{-2} + \alpha^{-4}} - \alpha^{-2}$. (Units are here $\hbar = m = e = 1$. $\alpha \equiv e^2/(\hbar c)$ is the fine structure constant.) All the other parts in the Hamiltonian are kept the same. The first result in this direction was obtained by T. \O stergaard S\o rensen \cite{OS}: If $E_Z$ denotes now the ground-state energy of a pseudo-relativistic neutral atom with nuclear charge $Z$, then,
\begin{gather}
\lim_{\substack{Z \to \infty\\ \alpha \to 0\\ Z\alpha \leq 2/\pi}}\frac{E_Z}{Z^{7/3}} = -C_{TF}.
\end{gather}
The reader will notice that, except for the fact that $\alpha$ is taken to go to zero while $Z\alpha$ is bounded by $2/\pi$, the result is the same as in the non-relativistic case. If $\alpha$ were not forced to go to zero, the underlying Hamiltonian would become unstable as soon as $Z\alpha$ exceeded $2/\pi$, since it is given by
\begin{gather}
\alpha^{-1}\Bigg[\sum_{i = 1}^Z\left(\sqrt{-\Delta_i + \alpha^{-2}} - \alpha^{-1} - \frac{Z\alpha}{|x_i|}\right) + \sum_{1 \leq i < j \leq Z}\frac{\alpha}{|x_ i - x_j|}\Bigg],
\end{gather}
and we note that $\sqrt{-\Delta + \alpha^{-2}} - \alpha^{-1} - Z\alpha /|x| \leq \sqrt{-\Delta} - Z\alpha / |x|$, which is unbounded from below for $Z\alpha > 2/\pi$. (This follows from the relativistic Hardy inequality: $|p| - 2/(\pi|x|) \geq 0$ and $|p| - c/|x|$ is unbounded from below for $c > 2/\pi$.)

We now proceed to study the first correction to the asymptotic result $-C_{TF}Z^{7/3}$, in the relativistic case. 
We will consider a neutral system of $M$ atoms of atomic numbers $(Z_1, Z_2, \ldots , Z_M)$ located at the distinct position vectors $(R_1, R_2, \ldots , R_M)$, and $Z \equiv \sum_{m = 1}^M Z_m$ pseudo-relativistic electrons. The system is then described by the following Hamiltonian
\begin{gather}
H_0 \equiv \sum_{m = 1}^Z\left(\sqrt{\alpha^{-2}p_m^2 + \alpha^{-4}} - \alpha^{-2}\right) - \sum_{m = 1}^Z\sum_{n = 1}^M\frac{Z_n}{|x_m - R_n|} + \sum_{m < n}\frac{1}{|x_m - x_n|}.
\end{gather}
If we denote by $E_0$ the infimum of $(\psi, H\psi)$ for all normalized $\psi$,
\begin{gather}
E_0 \equiv \text{inf spec }H_0,
\end{gather}
we are interested in the asymptotics of $E_0$ as the total nuclear charge $Z \to \infty$, while simultaneously the fine structure constant $\alpha \to 0$ and $\max_{1 \leq k \leq M}{Z_k\alpha} \leq 2/\pi$. 

We are interested in the Scott correction in this relativistic setting. 
For atoms this was obtained by Frank, Siedentop, and Warzel \cite{FSW} (see also e.g. \cite{FSW2,HS} for other related models).
The relativistic Scott correction for molecules was  proved by Solovej, \O stergaard S\o rensen, and Spitzer \cite{SSS}. The reader should at this point remember definitions \ref{definition.z} and \ref{definition.r} regarding $z$ and $r$. This time we shall use them backward: given vectors $z$ and $r$, $(Z_1, \ldots , Z_M)$ will be $Zz$ and $(R_1, \ldots , R_M)$ will be $Z^{-1/3}r$.
\begin{theorem}[Solovej, \O stergaard S\o rensen, Spitzer \cite{SSS}]
\label{relativistic.scott.correction}
There exists a continuous, non-increasing function $S_2 : [0, 2/\pi] \to \mathbb{R}$ with $S_2(0) = 1/4$ with the following property: Let $z = (z_1, \ldots, z_M) \in {\mathbb R}_{+}^M$ and $r = (r_1, \ldots, r_M) \in {\mathbb R}^{3M}$ be fixed given vectors and $r_0 > 0$ a fixed number such that $\sum_{i = 1}^M z_i = 1$ and $\min_{i \neq j}|r_i - r_j| > r_0$. Then, as $Z \to \infty$ and $\alpha \to 0$ while $\max_{1 \leq k \leq M} (Z z_k\alpha) \leq 2/\pi$,
\begin{gather}
E_0(Zz, Z^{-1/3}r) = Z^{7/3}E^{TF}(z, r) + 2\sum_{k = 1}^M Z_k^2 S_2(Z_k\alpha) + \mathcal{O}(Z^{2 - 1/30}),
\end{gather}
where $\left|\mathcal{O}(Z^{2 - 1/30})\right| \leq LZ^{2 - 1/30}$, $L$ depending only on $M$ and $r_0$.
\end{theorem}
The inclusion of (self-generated) magnetic fields was then considered by Erd\H{o}s, Fournais, and Solovej \cite{EFS}. 
Before stating their main result, we shall define the main Hamiltonian involved. The system interacts with a classical magnetic field $B$, coming from a magnetic vector potential $A$. (In this manner, $B = \nabla\times A$.) We shall assume that $A \in L^6(\mathbb{R}^3)$, $\nabla\otimes A \in L^2(\mathbb{R}^3)$, and $\nabla\cdot A = 0$. 
Here $\nabla \otimes A$ denotes the $3 \times 3$ matrix with entries $\partial_j A_k$ and therefore $|\nabla \otimes A|^2 = \sum_{j,k=1}^3 |\partial_j A_k|^2$.
Magnetic vector potentials $A$ with these properties will be called ``admissible.'' Note that the property $\nabla\cdot A = 0$ implies in particular that $\|\nabla\otimes A\|_2 = \|\nabla\times A\|_2$. The Hamiltonian in question is
\begin{align}
& \sum_{m = 1}^Z\mathcal{T}_m^{(\alpha)}(A) - \sum_{m = 1}^Z\sum_{n = 1}^M\frac{Z_n}{|x_m - R_n|} + \sum_{m < n}\frac{1}{|x_m - x_n|} + \frac{1}{8\pi\alpha^2}\int_{\mathbb{R}^3}|\nabla\times A|^2\,dx\nonumber\\
\equiv \, & H(A) + \frac{1}{8\pi\alpha^2}\int_{\mathbb{R}^3}\left|\nabla\times A\right|^2\,dx,
\end{align}
where $\mathcal{T}_m^{(\alpha)}(A) \equiv \sqrt{\alpha^{-2}\left[\sigma\cdot(p_m - A)\right]^2 + \alpha^{-4}} - \alpha^{-2}$. As we want to take into account the interaction of the electron spin with the magnetic field, the Pauli operator $\sigma\cdot(p - A)$ is used, as opposed to the magnetic momentum $p - A$. ($\sigma$ is the vector of Pauli matrices $(\sigma_1, \sigma_2, \sigma_3)$.) We are then interested in the infimum of $(\psi, H(A)\psi) + (8\pi\alpha^2)^{-1}\|B\|_2^2$ for all normalized $\psi$ and admissible magnetic vector potentials,
\begin{gather}
E \equiv \inf_A\left(\text{inf spec }H(A) + \frac{1}{8\pi\alpha^2}\int_{\mathbb{R}^3}\left|\nabla\times A\right|^2\,dx\right).
\end{gather}
We will now state the main result from Erd\H{o}s, Fournais, and Solovej, concerning the Scott correction in the pseudo-relavistic case, including magnetic fields and spin. The function $S_2$ appearing is exactly the same as in Theorem \ref{relativistic.scott.correction}.
\begin{theorem}[Erd\H{o}s, Fournais, Solovej \cite{EFS}]
\label{magnetic.relativistic.scott.correction}
Under the same assumptions as in Theorem \ref{relativistic.scott.correction}, let in addition $0 < \kappa_0 < 2/\pi$ be a fixed constant. As $Z \to \infty$ and $\alpha \to 0$ while $\max_k (Z z_k\alpha) \leq \kappa_0 < 2/\pi$,
\begin{gather}
E = Z^{7/3}E^{TF}(z, r) + 2\sum_{k = 1}^M Z_k^2 S_2(Z_k\alpha) + o(Z^2).
\end{gather}
\end{theorem}
The goal of this paper is to extend Theorem \ref{magnetic.relativistic.scott.correction} to the critical case $\max_k(Z_k\alpha) = 2/\pi$. 
It is natural to expect the domain of validity to be the same as in the case of Theorem~\ref{relativistic.scott.correction} (without magnetic field).
However, from a mathematical point of view this is delicate. In the critical case it is never allowed to use a fraction of the kinetic energy to control error terms. The solution in \cite{SSS} is to consider the `Hardy operator', i.e. combined kinetic+Coulomb terms, as a unit and use properties of this operator to control error terms. This is also the strategy in the present paper using recently obtained Hardy-Lieb-Thirring inequalities for the Pauli operator \cite{BF}.
We also want to point to another detail: It is convenient to carry out a localization in space to facilitate analysis. However, due to the square root, the localization of the kinetic energy is non-trivial in the relativistic setting. In \cite{SSS} the explicit integral kernel of the kinetic energy is used to obtain very explicit localization errors. In the Pauli case such formulae are not available, and even the tool of diamagnetic inequalities cannot be used. Therefore, proofs have to be based essentially only on the spectral theorem. This makes it likely that the methods developed in the present paper are applicable for many other problems.

The rest of the paper shall be dedicated to proving the following theorem. Again, $S_2$ is the function from Theorem \ref{relativistic.scott.correction}---allow us to stress that this means that the relativistic Scott-term is not sensitive to the self-generated magnetic field.
\begin{theorem}
\label{magnetic.relativistic.critical.scott.correction}
Under the same assumptions as in Theorem \ref{relativistic.scott.correction}, as $Z \to \infty$ and $\alpha \to 0$ while $\max_k (Z z_k \alpha)$ $\leq 2/\pi$,
\begin{gather}
E = Z^{7/3}E^{TF}(z, r) + 2\sum_{k = 1}^M Z_k^2 S_2(Z_k\alpha) + o(Z^2).
\end{gather}
\end{theorem}
The proof of Theorem \ref{magnetic.relativistic.critical.scott.correction} follows roughly the lines of \cite{EFS}, but with many important modifications. In particular, one of the main tools that allowed us to include the critical case $2/\pi$ is our ``Pauli-Hardy-Lieb-Thirring'' inequality from \cite{BF} (see Theorem~\ref{thm:CPHLT} below)
The next two sections will be dedicated to the proof of Theorem \ref{magnetic.relativistic.critical.scott.correction}.
\end{subsection}
\end{section}

\section{Analytical tools}
In this section we collect a number of tools that we will need in the proof.

First of all, the Scott term is characterized by the following.
\begin{lemma}[{\cite[Lemma 4.3]{SSS}}]
\label{lemma.2S.non.magnetic}
Let $0 < \alpha \leq 2/\pi$.
Let $\phi \in C_0^{\infty}({\mathbb R}^3)$ with $\sqrt{1-\phi^2} \in C^{\infty}({\mathbb R}^3)$ and such that $\phi(x) = 1$ for $|x|\leq 1$, $\phi(x) = 0$ for $|x| \geq 2$, and define, for $R>0$,
$\phi_R(x) := \phi(x/R)$.

Then,
\begin{align}
\lim_{R \to \infty}\,\sup_{0 < \alpha \leq 2/\pi}\left|\tr\left[\phi_R\left(\sqrt{\alpha^{-2}p^2 + \alpha^{-4}} - \alpha^{-2} - \frac{1}{|x|} \right)\phi_R\right]_- - I_R - 2S_2(\alpha)\right| = 0,
\end{align}
with
\begin{gather}\label{eq:IR}
I_R := \frac{2}{(2\pi)^3}\int_{\mathbb{R}^3}\!\int_{\mathbb{R}^3}\!\phi_R(x)^2\left(\frac{p^2}{2} - \frac{1}{|x|}\right)_-\,dx\,dp < 0.
\end{gather}
\end{lemma}

The following formula is very useful to deal with the non-locality of the square root operator.
\begin{proposition}
\label{proposition.pull.out}
(Pull-out formula.) Let $0 < a \leq 1$; $A_1, A_2, \ldots$ a collection of positive, self-adjoint operators; and $S_1, S_2, \ldots$ a collection of bounded, self-adjoint operators such that $\sum_{n = 1}^{\infty}S_n^2 = 1$. Then,
\begin{gather}
\left(\sum_{n = 1}^{\infty}S_n A_n S_n\right)^a \geq \sum_{n = 1}^{\infty}S_n A_n^a S_n.
\end{gather}
\end{proposition}
For a proof of Proposition \ref{proposition.pull.out}, the reader is referred to the comments surrounding 
\cite[Equation (1.12)]{BF}.

We also use the Critical Pauli-Hardy-Lieb-Thirring inequality below, which is a particular case of \cite[Theorem 1.1]{BF} or \cite[Theorem 1.2]{BF}. (We remark here that $x_- \equiv \min (x, 0)$).

\begin{theorem}[Critical Pauli-Hardy-Lieb-Thirring inequality]\label{thm:CPHLT}
For all electric potentials $V \in L^1_{\rm loc}({\mathbb R}^3)$ and all vector potentials $A \in H^1({\mathbb R}^3;{\mathbb R}^3)$ we have with $B = \nabla \times A$,
\begin{equation}
\label{inequality.pauli.hardy.lieb.thirring}
\tr\left(|\sigma\cdot (p - A)| - \frac{2}{\pi |x|} - V\right)_{-} \geq -C\left(\int_{\mathbb{R}^3}|B|^2\,dx + \int_{\mathbb{R}^3}V_+^4\,dx\right),
\end{equation}
for some universal constant $C>0$ independent of $A$ and $V$.
\end{theorem}

We also have a Daubechies type inequality for the Pauli operator. This is given in \cite[Theorem 2.2]{EFS},
\begin{theorem}[Daubechies inequality for Pauli operators]
\label{theorem.lieb.thirring.efs}
There is a $C > 0$ such that for every $\beta > 0$, potential $V$ with $V_+ \in L^{5/2}(\mathbb{R}^3)\cap L^4(\mathbb{R}^3)$, and magnetic fields $B = \nabla\times A \in L^2(\mathbb{R}^3)$, we have
\begin{align}
& \tr\left[\sqrt{\beta^{-2}[\sigma\cdot p_A]^2+ \beta^{-4}} - \beta^{-2} - V\right]_-\nonumber\\
\geq & -C\left[\int_{\mathbb{R}^3}V_+(x)^{5/2}\,dx + \beta^3\int_{\mathbb{R}^3}V_+(x)^4\,dx + \left(\int_{\mathbb{R}^3}B^2\,dx\right)^{3/4}\left(\int_{\mathbb{R}^3}V_+(x)^4\,dx\right)^{1/4}\right].
\end{align}
\end{theorem}

As in \cite{SSS} a Daubechies inequality (i.e. a relativistic inequality of Lieb-Thirring type) allowing for local Coulomb singularities will be an important ingredient in our analysis in the present paper. Our version allows for the magnetic Pauli-operator and is given in the next theorem.

\begin{theorem}[Magnetic Combined Daubechies-Lieb-Yau inequality]\label{thm:MCDLY}
Let $\alpha >0$, $R_1, \ldots, R_M \in {\mathbb R}^3$ and $W \in L^1_{\rm loc}({\mathbb R}^3)$ satisfying for some $0 \leq U \in L^{5/2} \cap L^4({\mathbb R}^3)$,
\begin{align}
W(x) \geq \sum_{j=1}^M - \frac{\nu}{d_R(x)} 1_{\{d(x,R_j) < \alpha\}} - U(x),
\end{align}
with $\alpha \nu \leq \frac{2}{\pi}$, and
\begin{align}\label{eq:SepCond}
\min_{k \neq \ell} | R_k - R_{\ell}| > (2+2\pi) \alpha.
\end{align}
Then,
\begin{align}\label{eq:CDLY}
\tr\left( \sqrt{\alpha^{-2}[\sigma\cdot p_A]^2 + \alpha^{-4}} - \alpha^{-2}+ W \right)_{-}
&\geq
- C   \nu^{5/2} \alpha^{1/2} - C \int U^{5/2}\,dx \nonumber \\
&\quad- C \alpha^3  \int U^{4}\,dx
- C \alpha^{-1}  \int B^2\,dx.
\end{align}
Furthermore, for sufficiently small values of $\nu \alpha$ ($\nu \alpha \leq (64 M)^{-1}$ will do), the constant term $-C   \nu^{5/2} \alpha^{1/2} $ can be omitted in \eqref{eq:CDLY}.
\end{theorem}

The proof of Theorem~\ref{thm:MCDLY} will be given in Section~\ref{sec:PfofDaubechies}.

\begin{remark}~
\begin{enumerate}
\item In the case $M=1$ we can take $R_1=0$ for simplicity of notation. In that case Theorem~\ref{thm:MCDLY} gives a Lieb-Thirring inequality for the (Pauli) relativistic kinetic energy operator with mass, i.e.  $\sqrt{\alpha^{-2}[\sigma\cdot p_A]^2 + \alpha^{-4}} - \alpha^{-2}$ and allowing for a critical Hardy singularity $\frac{2}{\pi \alpha |x|}$ near the origin.
\item In the non-magnetic case $A=0$, Theorem~\ref{thm:MCDLY} gives a somewhat improved version of \cite[Thm. 2.8]{SSS}.
\end{enumerate}
\end{remark}

\begin{section}{Proof of Theorem \ref{magnetic.relativistic.critical.scott.correction}}
\begin{subsection}{Preliminaries}
\label{subsection.preliminaries.proof.main.theorem}
We begin by noticing that it is possible by taking $A=0$, to obtain an upper bound for $E$ with the correct form,
\begin{gather}
E \leq \text{inf spec }H(0) = Z^{7/3}E^{TF}(z, r) + 2\sum_{k = 1}^M Z_k^2 S_2(Z_k\alpha) + o(Z^2),
\end{gather}
for $Z \to \infty$ and $\alpha \to 0$ while $\max_k Z_k\alpha \leq 2/\pi$, which follows directly from Theorem \ref{relativistic.scott.correction}. We shall then focus for the rest of the article on finding a lower bound with the proper form.

For the remainder of the section it will be assumed that $1/\pi \leq \max_{k}Z z_k\alpha \leq 2/\pi$. By appealing to Theorem~\ref{magnetic.relativistic.scott.correction}, this suffices for the proof of Theorem \ref{magnetic.relativistic.critical.scott.correction}.
We do not assume that $1/\pi \leq \max_{k}Z z_k\alpha \leq 2/\pi$ in Section \ref{section.continuation.proof}.
\end{subsection}
\begin{subsection}{First step}
We start with a correlation inequality as in the beginning of the proof of \cite[Theorem 1.1]{EFS}, using \cite[Equation (2.5)]{EFS} (which is based in turn on \cite[Theorem 2.9 and the calculation on page 55]{SSS}),
\begin{align}\label{eq:Corr}
\text{inf spec }H(A) \geq & \, \tr \left[\mathcal{T}^{(\alpha)}(A) - V_{Z, R}^{TF}(x) - CZ^{3/2}sg(x)\right]_- - D(\rho_{Z, R}^{TF}) - CsZ^{8/3} - Cs^{-1}Z,
\end{align}
where $\mathcal{T}^{(\alpha)}(A) \equiv \sqrt{\alpha^{-2}\left[\sigma\cdot (p - A)\right]^2 + \alpha^{-4}} - \alpha^{-2}$ (that is, $\mathcal{T}_m^{(\alpha)}$ for just one particle), $V_{Z, R}^{TF}$ is the Thomas-Fermi potential (see \eqref{eq:VTF}); $\rho_{Z, R}^{TF}$ is the Thomas-Fermi density; $s = Z^{-5/6}$;
\begin{gather}
D(f) \equiv \frac{1}{2}\int_{\mathbb{R}^3}\!\int_{\mathbb{R}^3}\frac{\overline{f(x)}f(y)}{|x - y|}\,dx\,dy;
\end{gather}
and $g$ is defined as
\begin{align}\label{eq:gfunction}
g(x) = 
\begin{cases}
(2s)^{-1/2} & \text{if  } \,\, d_R(x) < 2s,\\
d_R(x)^{-1/2} & \text{if } \,\, 2s \leq d_R(x) \leq Z^{-1/3},\\
0 & \text{if } \,\, Z^{-1/3} < d_R(x),
\end{cases}
\end{align}
where $d_R(x) \equiv \min_{1 \leq m \leq M}|x - R_m|$. The idea will be now to effectively eliminate $CZ^{3/2}sg$ from
\begin{gather}
\tr\left[\mathcal{T}^{(\alpha)}(A) - V_{Z, R}^{TF}(x) - CZ^{3/2}sg(x)\right]_-,
\end{gather}
in the following sense,
\begin{theorem}
\label{theorem.trace.splitting}
\begin{align}
\tr\left[\mathcal{T}^{(\alpha)}(A) - V_{Z, R}^{TF}(x) - CZ^{3/2}sg(x)\right]_{-} \geq & \left(1 - Z^{-1/2}\right)\tr\left[\mathcal{T}^{(\alpha)}(A) - V_{Z, R}^{TF}(x)\right]_-\nonumber\\
& \, -CZ^{-1/2} \alpha^{-1} \int|\nabla\times A|^2\,dx + o(Z^2).
\label{equation.theorem.trace.splitting}
\end{align}
\end{theorem}

\begin{proof}[Proof of Theorem \ref{theorem.trace.splitting}]
Recall that $s= Z^{-5/6}$ and the definition of $g$ in \eqref{eq:gfunction}.
The theorem clearly follows from the estimate
\begin{align}\label{eq:EstimError}
\varepsilon \tr\left[\mathcal{T}^{(\alpha)}(A) - V_{Z, R}^{TF}(x) - C \varepsilon^{-1} Z^{3/2}sg(x)\right]_{-} \geq &  -C\varepsilon \alpha^{-1} \int|\nabla\times A|^2\,dx - C \varepsilon Z^{7/3},
\end{align}
with $\varepsilon = Z^{-1/2}$.
The estimate \eqref{eq:EstimError} is a direct application of the magnetic combined Daubechies-Lieb-Yau inequality, Theorem~\ref{thm:MCDLY} (see also \cite[p.56]{SSS} where the similar term is estimated in the non-magnetic situation).
\end{proof}

After a scaling argument in Subsection \ref{subsection.scaling}, the right hand side of \eqref{equation.theorem.trace.splitting}---with the magnetic field energy added---will in Subsection \ref{subsection.semiclassical.theorem} be reduced to a semiclassical problem. This semiclassical problem is analyzed in Section \ref{section.continuation.proof}.
\end{subsection}

\begin{subsection}{Scaling}
\label{subsection.scaling}
Using Theorem \ref{theorem.trace.splitting}, we conclude from \eqref{eq:Corr} that, for $\alpha$ small enough
\begin{align}
\label{estimate.gse.stage.1}
\text{inf spec }H(A) + \frac{1}{8\pi\alpha^2}\int|\nabla\times A|^2\,dx \geq & \, \left(1 - Z^{-1/2}\right)\tr\left[\mathcal{T}^{(\alpha)}(A) - V_{Z, R}^{TF}(x)\right]_- - D(\rho_{Z, R}^{TF})\nonumber\\
& \, + \left(\frac{1}{8\pi\alpha^2} - CZ^{-1/2} \alpha^{-1} \right)\int|\nabla\times A|^2\,dx + o(Z^2)\nonumber\\
\geq & \, \left(1 - Z^{-1/2}\right)\tr\left[\mathcal{T}^{(\alpha)}(A) - V_{Z, R}^{TF}(x)\right]_- - D(\rho_{Z, R}^{TF})\nonumber\\
& \, + \frac{1}{16\pi\alpha^2}\int|\nabla\times A|^2\,dx + o(Z^2).
\end{align}
\begin{remark}
This corresponds to Equation (2.9) in \cite{EFS}, with an important exception: the first term on the right side of that equation is
\begin{gather}
\label{trace.hamiltonian.factor.bigger.than.1}
\left(1 - Z^{-1/2}\right)\tr\left[\mathcal{T}^{(\alpha)}(A) - \left(1 - Z^{-1/2}\right)^{-1}V_{Z, R}^{TF}(x)\right]_-,
\end{gather}
which is not good enough for our purposes, because at criticality ($\max_{1 \leq k \leq M}Z_k\alpha = 2/\pi$) the Hamiltonian inside \eqref{trace.hamiltonian.factor.bigger.than.1} is unstable for $A = 0$ (and therefore when the infimum over $A$ is considered), because of the following reasons: $(\sigma\cdot p)^2 = p^2$, $\sqrt{\alpha^{-2}p^2 + \alpha^{-4}} - \alpha^{-2}$ is less than or equal to $\alpha^{-1}|p|$, $V^{TF}_{Z, R}$ behaves like $Z_k/|x - R_k|$ close to the nucleus at $R_k$ (Equation \eqref{estimate.VTF.Coulomb}), and $(1 - Z^{-1/2})^{-1} > 1$ for all $Z > 1$ (recall the relativistic Hardy inequality mentioned in the introduction). Theorem \ref{theorem.trace.splitting} is what allows us to circumvent this problem.
\end{remark}
We will now apply the following semiclassical scaling, as it appears in subsection II.B from \cite{EFS}. We define
\begin{align}
\kappa \equiv & \, \min_{1 \leq k \leq M}\frac{2}{\pi z_k},\\
h \equiv & \, \kappa^{1/2}Z^{-1/3},\\
\beta \equiv & \, Z^{2/3}\alpha\kappa^{-1/2} = Z\alpha h\kappa^{-1}, \label{eq:defBeta}\\
T_h(A) \equiv & \, \left[\sigma\cdot(-ih\nabla + A)\right]^2.
\end{align}
We replace the potential $A$ by
\begin{gather}
\widetilde{A}(x) \equiv Z^{-2/3}\kappa^{1/2}A(Z^{-1/3}x),
\end{gather}
and the scaling properties \eqref{eq:TF-scaling} of Thomas-Fermi theory,
with $a = Z^{1/3}$, we get, from \eqref{estimate.gse.stage.1},
\begin{align}
& \, \text{inf spec }H(A) + \frac{1}{8\pi\alpha^2}\int|\nabla\times A|^2\,dx\nonumber\\
\geq & \, Z^{4/3}\kappa^{-1}\left(1 - Z^{-1/2}\right)\left[\tr\left(\sqrt{\beta^{-2}T_h(\widetilde{A}) + \beta^{-4}} - \beta^{-2} - \kappa V_{z, r}^{TF}\right)_- + \frac{\kappa^{1/2}}{16\pi\beta^2 h^3}\int|\nabla\otimes \widetilde{A}|^2\,dx\right]\nonumber\\
& \, -Z^{7/3}D(\rho_{z, r}^{TF}) + o(Z^2).
\label{equation.preliminary.estimate.for.main.theorem}
\end{align}

\end{subsection}

\begin{subsection}{Semiclassical theorem}
\label{subsection.semiclassical.theorem}
As in \cite[Subsection II.B]{EFS}, the idea will be now to use the following theorem (where the extension below compared to \cite{EFS} is that we allow criticality, i.e. $\widetilde \kappa z = 2/\pi$), that we prove in the next section, in order to finish the proof of Theorem \ref{magnetic.relativistic.critical.scott.correction},
\begin{theorem}
\label{theorem.preliminary.main.result}
Let $\lambda > 0$ and $0 < \widetilde{\kappa}\max\left\{z_1, \ldots, z_M\right\} \leq 2/\pi$. There is a function $\xi : \mathbb{R}_+ \to \mathbb{R}_+$ with $\xi(t) \to 0$ as $t \to 0$ such that if $0 < \beta \leq h$, then
\begin{multline}
\label{equation.preliminary.main.result}
\left|\inf_{\widetilde{A}}\left\{\tr\left[\sqrt{\beta^{-2}T_h(\widetilde{A}) + \beta^{-4}} - \beta^{-2} - \widetilde{\kappa}V_{z, r}^{TF}\right]_- + \frac{\lambda}{\beta^2 h^3}\int_{\mathbb{R}^3}\left|\nabla\otimes\widetilde{A}\right|^2\right\}\right.\\
\left. - \frac{2}{(2\pi h)^3}\int_{\mathbb{R}^3}\int_{\mathbb{R}^3}\left[\frac{p^2}{2} - \widetilde{\kappa}V_{z, r}^{TF}(x)\right]_{-}\,dx\,dp - \frac{2}{h^2}\sum_{k = 1}^M (z_k\widetilde{\kappa})^2 S_2\left(\beta h^{-1}\widetilde{\kappa}z_k\right)\right| \leq \frac{\xi(h)}{h^2}.
\end{multline}
\end{theorem}

Before discussing the proof of Theorem \ref{theorem.preliminary.main.result}, we shall explain how our main Theorem \ref{magnetic.relativistic.critical.scott.correction} follows directly from Theorem \ref{theorem.preliminary.main.result}.

\begin{proof}[Proof of Theorem \ref{magnetic.relativistic.critical.scott.correction}]
We explained in subsection \ref{subsection.preliminaries.proof.main.theorem} that all we really need is a lower bound for $E$, since an upper bound follows from the non-magnetic case. In order to find a lower bound we first notice that, from the estimate \eqref{equation.preliminary.estimate.for.main.theorem} and Theorem \ref{theorem.preliminary.main.result} with $\lambda = \kappa^{1/2}/(16\pi)$,
\begin{align}
\text{inf spec }H(A) &+ \frac{1}{8\pi\alpha^2}\int|\nabla\times A|^2\,dx
\nonumber\\
\geq & \, Z^{4/3}\kappa^{-1}(1 - Z^{-1/2})\left\{\frac{2}{(2\pi h)^3}\int\!\!\!\int\left[\frac{p^2}{2} - \kappa V^{TF}_{z, r}(x)\right]_-\,dx\,dp\right.\nonumber\\
& \, \left. + 2h^{-2}\sum_{k = 1}^M (z_k \kappa)^2 S_2(\beta h^{-1}z_k \kappa) - \frac{\xi(h)}{h^2}\right\} - Z^{7/3}D(\rho_{z, r}^{TF}) + o(Z^2).
\label{equation.preliminary.estimate.for.main.theorem.second}
\end{align}
The $p$-part of the first integral can be calculated, yielding
\begin{gather}
\int\!\!\!\int\left[\frac{p^2}{2} - \kappa V^{TF}_{z, r}(x)\right]_-\,dx\,dp = -C_0\kappa^{5/2}\int V_{z, r}^{TF}(x)^{5/2}\,dx,
\end{gather}
and then, using \eqref{equation.semiclassical.V.TF},
\begin{align}
& \, Z^{4/3}\kappa^{-1}\left(1 - Z^{-1/2}\right)\frac{2}{(2\pi h)^3}\int\!\!\!\int\left[\frac{p^2}{2} - \kappa V_{z, r}^{TF}(x)\right]_-\,dx\,dp - Z^{7/3}D(\rho_{z, r}^{TF})\nonumber\\
= & \, Z^{7/3}\left[-\frac{2C_0}{(2\pi)^3}\int V_{z, r}^{TF}(x)^{5/2}\,dx - D(\rho_{z, r}^{TF})\right] + o(Z^2)\nonumber\\
= & \, Z^{7/3}E(z, r) + o(Z^2),
\end{align}
and therefore
\begin{align}
\text{inf spec }H(A) + \frac{1}{8\pi\alpha^2}\int|\nabla\times A|^2\,dx \geq Z^{7/3}E(z, r) + 2\sum_{k = 1}^M Z_k^2 S_2(Z_k\alpha) + o(Z^2),
\end{align}
which proves the lower bound. This concludes the proof of Theorem \ref{magnetic.relativistic.critical.scott.correction}.
\end{proof}

The only step that is pending then is the proof of Theorem \ref{theorem.preliminary.main.result}. This we shall finish in Section~\ref{section.continuation.proof}.
\end{subsection}
\end{section}

\begin{section}{Proof of Theorem~\ref{theorem.preliminary.main.result}}
\label{section.continuation.proof}
\begin{subsection}{Preliminaries}
Let $\mathcal{M}(\beta, h, \widetilde{\kappa}, \lambda)$ be the quantity inside the absolute value in Equation \eqref{equation.preliminary.main.result}. The goal is to show that there is a $\xi : \mathbb{R}_+ \to \mathbb{R}_+$ with $\lim_{t \to 0}\xi(t) = 0$ such that
\begin{gather}
-\frac{\xi(h)}{h^2} \leq \mathcal{M}(\beta, h, \widetilde{\kappa}, \lambda) \leq \frac{\xi(h)}{h^2}.
\end{gather}
For the upper bound, we set $\widetilde{A} = 0$ and get
\begin{align}
\mathcal{M}\left(\beta, h, \widetilde{\kappa}, \lambda\right) \leq & \, \tr\left[\sqrt{\beta^{-2}h^2p^2 + \beta^{-4}} - \beta^{-2} - \widetilde{\kappa}V_{z, r}^{TF}\right]_-\nonumber\\
& \, - \frac{2}{(2\pi h)^3}\int\!\!\!\int\left[\frac{p^2}{2} - \widetilde{\kappa}V_{z, r}^{TF}(x)\right]_-\,dx\,dp - \frac{2}{h^2}\sum_{k = 1}^M (z_k\widetilde{\kappa})^2 S_2\left(\beta h^{-1}\widetilde{\kappa}z_k\right),
\end{align}
and this quantity is less than or equal to $Ch^{1/10}/h^2$, by \cite[Theorem 1.4]{SSS}. (The statement of that theorem is precisely our statement of Theorem \ref{theorem.preliminary.main.result} with $\widetilde{A} = 0$ and $\xi(h) = Ch^{1/10}$.)

The only remaining part then is to prove the lower bound $\mathcal{M} \geq - \xi(h)/h^2$. We shall follow the arguments in \cite[sections III, IV, and V]{EFS}, with modifications in order to include the critical constant $2/\pi$. 
We consider a smooth partition of unity $\theta_{+}, \theta_{-}$ on $[0, \infty)$ with 
\begin{gather}
\theta_-^2 + \theta_+^2 = 1,
\end{gather}
and $\theta_{-}(t) = 1$ for $t < 1$, and equal to 0 if $t > 2$. 

Let $u, U > 0$ and define
\begin{align}\label{eq:Phiphi}
\Phi_{\pm}(x) & = \, \theta_{\pm}\left(d_r(x)/U\right),\\
\phi_{\pm}(x) & = \, \theta_{\pm}\left(d_r(x)/u\right).
\end{align}
We let $u$ be small enough so that, if $\theta_{u, k}(x) \equiv \theta_-(|x - r_k|/u)$, then
\begin{gather}
\text{supp }\theta_{u, i}\cap \text{supp }\theta_{u, j} = \emptyset,
\end{gather}
for $i \neq j$. Therefore,
\begin{gather}
\phi_-(x) = \sum_{k = 1}^M\theta_{u, k}(x).
\end{gather}
Furthermore, we let $U$ be large enough so that
\begin{gather}
\bigcup_{k = 1}^M\text{supp }\theta_{u, k} \subset \left\{x \in \mathbb{R}^3 : \Phi_-(x) = 1\right\};
\end{gather}
with this,
\begin{gather}
1 = \phi_-^2 + \phi_+^2 = \sum_{k = 1}^M\theta_{u, k}^2 + \phi_+^2 = \sum_{k = 1}^M\theta_{u, k}^2 + \phi_+^2\Phi_-^2 + \phi_+^2\Phi_+^2 = \sum_{k = 1}^M\theta_{u, k}^2 + \phi_+^2\Phi_-^2 + \Phi_+^2.
\end{gather}
We will fix at this point $u = h^{2 - \delta}$, $\delta = 1/11$, and $U = h^{-1}$. For $h$ small enough, the above properties will be satisfied. We then have, by the IMS formula,
\begin{align}
T_h(\widetilde{A}) + \beta^{-2} \geq & \, \sum_{k = 1}^M\theta_{u, k}\left(T_h(\widetilde{A}) + \beta^{-2} - E\right)\theta_{u, k} + \Phi_-\phi_+\left(T_h(\widetilde{A}) + \beta^{-2} - E\right)\Phi_-\phi_+\nonumber\\
& \, + \Phi_+\left(T_h(\widetilde{A}) + \beta^{-2} - E\right)\Phi_+,
\end{align}
where $E$ is the IMS localization error. One can easily verify that, with $d(x) =: d_r(x)$,
\begin{align}
E\theta_{u, k}^2 & \leq \, Ch^2u^{-2},\\
E\Phi_-^2\phi_+^2 & \leq \, Ch^2\left(u^{-2}1_{\left\{u \leq d(x) \leq 2u\right\}} + U^{-2}1_{\left\{U \leq d(x) \leq 2U\right\}}\right) \equiv Ch^2 W_{u, U}(x),\\
E\Phi_+^2 & \leq \, Ch^2 U^{-2}1_{\left\{d(x) \leq 2U\right\}}.
\end{align}
We notice that each of these expressions is smaller than $\beta^{-2}$ for small enough $h$
(recall from Theorem~\ref{theorem.preliminary.main.result} that $\beta\leq h$ by assumption). We then have, by the pull-out estimate (Proposition \ref{proposition.pull.out}),
\begin{align}
& \, \sqrt{\beta^{-2}T_h(\widetilde{A}) + \beta^{-4}} - \beta^{-2} - \widetilde{\kappa}V_{z, r}^{TF}\nonumber\\
\geq & \, \sum_{k = 1}^M\theta_{u, k}\left(\sqrt{\beta^{-2}T_h(\widetilde{A}) + \beta^{-4} - C\beta^{-2}h^2 u^{-2}} - \beta^{-2} - \widetilde{\kappa}V_{z, r}^{TF}\right)\theta_{u, k}\nonumber\\
& \, + \Phi_-\phi_+\left(\sqrt{\beta^{-2}T_h(\widetilde{A}) + \beta^{-4} - C\beta^{-2}h^2W_{u, U}} - \beta^{-2} - \widetilde{\kappa}V_{z, r}^{TF}\right)\phi_+\Phi_-\nonumber\\
& \, + \Phi_+\left(\sqrt{\beta^{-2}T_h(\widetilde{A}) + \beta^{-4} - Ch^2\beta^{-2}U^{-2}1_{\left\{d(x) \leq 2U\right\}}} - \beta^{-2} - \widetilde{\kappa}V_{z, r}^{TF}\right)\Phi_+.
\end{align}
The trace of $\sqrt{\beta^{-2}T_h(\widetilde{A}) + \beta^{-4}} - \beta^{-2} - \widetilde{\kappa}V_{z, r}^{TF}$ can therefore be split into three pieces. 

In conclusion, we have (for arbitrary $\widetilde{A}$ and where in ${\mathcal T}_1$ we have dropped the corresponding part of the phase space integral for a lower bound)
\begin{align}\label{eq:SplittingInTs}
&\tr\left[\sqrt{\beta^{-2}T_h(\widetilde{A}) + \beta^{-4}} - \beta^{-2} - \widetilde{\kappa}V_{z, r}^{TF}\right]_- + \frac{\lambda}{\beta^2 h^3}\int_{\mathbb{R}^3}\left|\nabla\otimes\widetilde{A}\right|^2 \nonumber \\
&\quad- \frac{2}{(2\pi h)^3}\int_{\mathbb{R}^3}\int_{\mathbb{R}^3}\left[\frac{p^2}{2} - \widetilde{\kappa}V_{z, r}^{TF}(x)\right]_{-}\,dx\,dp - \frac{2}{h^2}\sum_{k = 1}^M (z_k\widetilde{\kappa})^2 S_2\left(\beta h^{-1}\widetilde{\kappa}z_k\right) \nonumber \\
&\geq
{\mathcal T}_1 + {\mathcal T}_2 + \sum_{k=1}^M {\mathcal T}_{3,k},
\end{align}
where
\begin{align}\label{eq:Ts}
{\mathcal T}_1 &:= \tr \Big[ \Phi_+\left(\sqrt{\beta^{-2}T_h(\widetilde{A}) + \beta^{-4} - Ch^2\beta^{-2}U^{-2}1_{\left\{d(x) \leq 2U\right\}}} - \beta^{-2} - \widetilde{\kappa}V_{z, r}^{TF}\right)\Phi_+ \Big]_{-} \nonumber \\
&\quad + \frac{\lambda}{3\beta^2 h^3}\int_{\mathbb{R}^3}\left|\nabla\otimes\widetilde{A}\right|^2,
\nonumber \\
{\mathcal T}_2 &:=  \tr \Big[  \Phi_-\phi_+\left(\sqrt{\beta^{-2}T_h(\widetilde{A}) + \beta^{-4} - C\beta^{-2}h^2W_{u, U}} - \beta^{-2} - \widetilde{\kappa}V_{z, r}^{TF}\right)\phi_+\Phi_- \Big]_{-} 
+ \frac{\lambda}{3\beta^2 h^3}\int_{\mathbb{R}^3}\left|\nabla\otimes\widetilde{A}\right|^2\nonumber \\
&\qquad- \frac{2}{(2\pi h)^3}\int_{\mathbb{R}^3}\int_{\mathbb{R}^3}  \Phi_{-}^2(x)\phi_{+}^2(x)
\left[\frac{p^2}{2} - \widetilde{\kappa}V_{z, r}^{TF}(x)\right]_{-}\,dx\,dp, \nonumber \\
{\mathcal T}_{3,k} &:=  \tr \Big[ \theta_{u, k}\left(\sqrt{\beta^{-2}T_h(\widetilde{A}) + \beta^{-4} - C\beta^{-2}h^2 u^{-2}} - \beta^{-2} - \widetilde{\kappa}V_{z, r}^{TF}\right)\theta_{u, k} \Big]_{-}
+ \frac{\lambda}{3 M \beta^2 h^3}\int_{\mathbb{R}^3}\left|\nabla\otimes\widetilde{A}\right|^2\nonumber \\
&\qquad- \frac{2}{(2\pi h)^3}\int_{\mathbb{R}^3}\int_{\mathbb{R}^3}  \theta_{u, k}^2(x)
\left[\frac{p^2}{2} - \widetilde{\kappa}V_{z, r}^{TF}(x)\right]_{-}\,dx\,dp- \frac{2}{h^2} (z_k\widetilde{\kappa})^2 S_2\left(\beta h^{-1}\widetilde{\kappa}z_k\right).
\end{align}

The first two parts, ${\mathcal T}_1$ and ${\mathcal T}_2$, are localized away from the nuclei, and for this reason the exact value of the coupling constant $\widetilde{\kappa}\max(z_1, \ldots , z_M)$ is irrelevant in these regions. Because of this, in order to bound ${\mathcal T}_1$ and ${\mathcal T}_2$ we shall merely use the following two results, obtained in \cite[subsections III. A. and III. B]{EFS}. The proofs in \cite{EFS} also work for $\widetilde{\kappa}\max(z_1, \ldots , z_M) = 2/\pi$.
\begin{theorem}[Erd\H{o}s, Fournais, Solovej]\label{thm:EFS_far}
For sufficiently small $h > 0$, $0 < \beta \leq h$, $0 < \widetilde{\kappa}\max(z_1, \ldots , z_M)$ $\leq 2/\pi$, and any admissible $\widetilde{A}$,
\begin{align}
& \, \tr\left[\Phi_+\left(\sqrt{\beta^{-2}T_h(\widetilde{A}) + \beta^{-4} - Ch^2\beta^{-2}U^{-2}1_{\left\{d(x) \leq 2U\right\}}} - \beta^{-2} - \widetilde{\kappa}V_{z, r}^{TF}\right)\Phi_+\right]_-\nonumber\\
\geq & \, -C - Ch^{-1}\int_{\mathbb{R}^3}|\nabla\otimes\widetilde{A}|^2\,dx.
\end{align}
\end{theorem}
\begin{theorem}[Erd\H{o}s, Fournais, Solovej]\label{thm:EFS_semiclassics}
For sufficiently small $h > 0$, $0 < \beta \leq h$, $0 < \widetilde{\kappa}\max(z_1, \ldots , z_M)$ $\leq 2/\pi$, $\widetilde{\lambda} > 0$, and any admissible $\widetilde{A}$,
\begin{align}\label{eq:EstimateT2}
& \, \tr\left[\Phi_-\phi_+\left(\sqrt{\beta^{-2}T_h(\widetilde{A}) + \beta^{-4} - C\beta^{-2}h^2W_{u, U}} - \beta^{-2} - \widetilde{\kappa}V_{z, r}^{TF}\right)\phi_+\Phi_-\right]_- + \frac{\widetilde{\lambda}}{\beta^2 h^3}\int|\nabla\otimes\widetilde{A}|^2\,dx\nonumber\\
\geq & \, \frac{2}{(2\pi h)^3}\int\!\!\!\int \Phi_-^2(x)\phi_+(x)^2\left[\frac{p^2}{2} - \widetilde{\kappa}V_{z, r}^{TF}(x)\right]_-\,dx\,dp - Ch^{-2 + \delta/22}.
\end{align}
\end{theorem}

\begin{remark}
Notice that the error term in \eqref{eq:EstimateT2} is different from \cite[(3.25)]{EFS}. This is because we need to choose the length scale $u=h^{2-\delta}$ in \eqref{eq:Phiphi} much shorter than the corresponding choice ($h^{3/2}$) in \cite{EFS}. However, the proof from \cite{EFS} works as long as $u \gg h^2$.
\end{remark}

The only part remaining then is the one localized close to the nuclei, i.e. the ${\mathcal T}_{3,k}$'s.
\end{subsection}
\begin{subsection}{The region close to the nuclei}
We will follow roughly the arguments in \cite[III.C]{EFS}, but with several important changes. We will study an indiviual term ${\mathcal T}_{3,k}$. We assume without loss of generality that $r_k = 0$. The expression we will be working on is therefore,
\begin{align}
\label{equation.first.close.to.the.nuclei}
&\tr\left[\theta_-\left(\frac{|x|}{u}\right)\left(\sqrt{\beta^{-2}T_h(\widetilde{A}) + \beta^{-4} - C\beta^{-2}h^2u^{-2}} - \beta^{-2} - \widetilde{\kappa}V_{z, r}^{TF}\right)\theta_-\left(\frac{|x|}{u}\right)\right]_-
+ \frac{\lambda}{3 M \beta^2 h^3}\int_{\mathbb{R}^3}\left|\nabla\otimes\widetilde{A}\right|^2\nonumber \\
&\qquad- \frac{2}{(2\pi h)^3}\int_{\mathbb{R}^3}\int_{\mathbb{R}^3} \theta_{-}\left(\frac{|x|}{u}\right)^2(x)
\left[\frac{p^2}{2} - \widetilde{\kappa}V_{z, r}^{TF}(x)\right]_{-}\,dx\,dp- \frac{2}{h^2} (z_k\widetilde{\kappa})^2 S_2\left(\beta h^{-1}\widetilde{\kappa}z_k\right).
\end{align}
First we study the phase space integral in \eqref{equation.first.close.to.the.nuclei}.
We notice that $V_{z, r}^{TF}$ may be replaced by $z_k/|x|$ at a negligible cost: by means of equation \eqref{estimate.VTF.Coulomb} and the expansion for $-1 < x < 1$ given by $(1 + x)^{5/2} = 1 + C(1 + \xi)^{3/2}x$, $\xi$ in the closed interval generated by 0 and $x$, we find
\begin{align}
& \left|\frac{2}{(2\pi h)^3}\int\!\!\!\int\theta_{u, k}^2(x)\left[\left(\frac{p^2}{2} - \frac{\widetilde{\kappa}z_k}{|x|}\right)_- - \left(\frac{p^2}{2} - \widetilde{\kappa}V_{z, r}^{TF}(x)\right)_-\right]\,dx\,dp\right|\nonumber\\
\leq & \, C\frac{2}{(2\pi h)^3}\int\theta_{u, k}^2(x)\left|\left(\frac{z_k}{|x|}\right)^{5/2} - V_{z, r}^{TF}(x)^{5/2}\right|\,dx\nonumber\\
\leq & \, C\frac{2z_k^{3/2}}{(2\pi h)^3}\int\frac{\theta_{u, k}^2(x)}{|x|^{3/2}}\,dx = Ch^{-3}u^{3/2} \leq Ch^{-3/22}.
\end{align}
Also, a simple change of variables shows that
\begin{align}
\frac{2}{(2\pi h)^3}\int\!\!\!\int\theta_{u, k}^2(x)\left(\frac{p^2}{2} - \frac{\widetilde{\kappa}z_k}{|x|}\right)_-\,dx\,dp 
&= \frac{(\widetilde{\kappa}z_k)^2}{h^2}\frac{2}{(2\pi)^3}\int\!\!\!\int\theta_-^2\left(|x|/\mathcal{R}\right)\left(\frac{p^2}{2} - \frac{1}{|x|}\right)_-\,dx\,dp,
\end{align}
with $\mathcal{R} := \widetilde{\kappa}z_k h^{-2} u$.
Therefore, our final estimate on the phase space integral is, 
\begin{align}\label{eq:PhaseSpace}
 -\frac{2}{(2\pi h)^3}\int\!\!\!\int & \theta_{u, k}^2(x)\left(\frac{p^2}{2} - \frac{\widetilde{\kappa}z_k}{|x|}\right)_- \,dx\,dp\nonumber \\
 &\geq
 - \frac{(\widetilde{\kappa}z_k)^2}{h^2}\frac{2}{(2\pi)^3}\int\!\!\!\int\theta_-^2\left(|x|/\mathcal{R}\right)\left(\frac{p^2}{2} - \frac{1}{|x|}\right)_-\,dx\,dp - C h^{-3/22}.
\end{align}

Since
$\xi \mapsto \sqrt{a^2 + \xi^2} -\xi$ in non-increasing on $[0,\infty)$ for all real $a$, we get that the trace in \eqref{equation.first.close.to.the.nuclei} is bounded from below by
\begin{gather}
\label{equation.second.close.to.the.nuclei}
\tr\left[\theta_-\left(|x|/u\right)\left(\sqrt{\beta^{-2}T_h(\widetilde{A}) + \beta^{-4}} - \beta^{-2} - \widetilde{\kappa}V_{z, r}^{TF} - Ch^2u^{-2}\right)\theta_{-}\left(|x|/u\right)\right]_-.
\end{gather}
We use the fact that by \eqref{estimate.VTF.Coulomb} on the support of $\theta_{-}\left(|\cdot|/u\right)$,
\begin{gather}
V_{z, r}^{TF} \leq C + \frac{z_k}{|x|},
\end{gather}
for a constant $C$, and we see that the trace in \eqref{equation.second.close.to.the.nuclei} is bounded from below by
\begin{gather}
\tr\left[\theta_-\left(|x|/u\right)\left(\sqrt{\beta^{-2}T_h(\widetilde{A}) + \beta^{-4}} - \beta^{-2} - \frac{\widetilde{\kappa}z_k}{|x|} - C\left(h^2u^{-2} + 1\right)\right)\theta_-\left(|x|/u\right)\right]_-.
\end{gather}
Now, with the scaling $x = h^2 y (\widetilde{\kappa}z_k)^{-1}$, this is greater than or equal to, with $T :=T_{h=1}$,
\begin{gather}
\label{equation.third.close.to.the.nuclei}
\frac{(\widetilde{\kappa}z_k)^2}{h^2}\tr\left[\theta_-\left(|y|/\mathcal{R}\right)\left(\sqrt{\widetilde{\alpha}^{-2}T(\overline{A}) + \widetilde{\alpha}^{-4}} - \widetilde{\alpha}^{-2} - \frac{1}{|y|} - C\left(h^4u^{-2} + h^2\right)\right)\theta_-\left(|y|/\mathcal{R}\right)\right]_-,
\end{gather}
with
\begin{align}\label{eq:params}
\mathcal{R} := \widetilde{\kappa}z_k u h^{-2}, \qquad \widetilde{\alpha} := \widetilde{\kappa}z_k\beta h^{-1}, \quad \text{ and } \quad \overline{A}(y) := h(\widetilde{\kappa}z_k)^{-1}\widetilde{A}\left(h^2 y(\widetilde{\kappa}z_k)^{-1}\right).
\end{align}
By recalling that $u = h^{2 - 1/11}$, we can finally bound the expression in \eqref{equation.third.close.to.the.nuclei} from below by
\begin{gather}
\label{equation.fourth.close.to.the.nuclei}
\frac{\left(\widetilde{\kappa}z_k\right)^2}{h^2}\tr\left[\theta_-\left(|y|/\mathcal{R}\right)\left(\sqrt{\widetilde{\alpha}^{-2}T(\overline{A}) + \widetilde{\alpha}^{-4}} - \widetilde{\alpha}^{-2} - \frac{1}{|y|} - C\mathcal{R}^{-2}\right)\theta_-\left(|y|/\mathcal{R}\right)\right]_-.
\end{gather}
In conclusion,
\begin{align}\label{eq:Traces}
&\tr\left[\theta_-\left(|x|/u\right)\left(\sqrt{\beta^{-2}T_h(\widetilde{A}) + \beta^{-4} - C\beta^{-2}h^2u^{-2}} - \beta^{-2} - \widetilde{\kappa}V_{z, r}^{TF}\right)\theta_-\left(|x|/u\right)\right]_- \nonumber \\
& \geq 
\frac{\left(\widetilde{\kappa}z_k\right)^2}{h^2}\tr\left[\theta_-\left(|y|/\mathcal{R}\right)\left(\sqrt{\widetilde{\alpha}^{-2}T(\overline{A}) + \widetilde{\alpha}^{-4}} - \widetilde{\alpha}^{-2} - \frac{1}{|y|} - C\mathcal{R}^{-2}\right)\theta_-\left(|y|/\mathcal{R}\right)\right]_-.
\end{align}
Notice also that
\begin{align}\label{eq:ChngVarField}
\int |\nabla \otimes \overline{A}|^2 = (\widetilde{\kappa} z_k)^{-1} \int |\nabla \otimes \widetilde{A}|^2.
\end{align}
Combining \eqref{eq:PhaseSpace}, \eqref{eq:Traces} and \eqref{eq:ChngVarField}, 
and introducing the notation $E(R,\alpha, \Lambda)$ from Theorem~\ref{theorem.asymptotics.E} below, 
we get (still with the convention $r_k=0$ and with parameters from \eqref{eq:params}),
\begin{align}\label{eq:CombinedT3k}
{\mathcal T}_{3,k} \geq \frac{\widetilde{\kappa}^2 z_k^2}{h^2} \left( E\left(\mathcal{R}, \widetilde{\alpha}, \frac{\lambda}{3M \beta^2 h \widetilde{\kappa} z_k}\right) - 2 S_2(\widetilde{\alpha})\right)
-C h^{-3/50}.
\end{align}
Notice that by the choice of $u$ (and since $\beta \leq h$) it is clear that 
$$\mathcal{R}^{11} \frac{3M \beta^2 h \widetilde{\kappa} z_k}{\lambda} \rightarrow 0
$$ 
as $h \rightarrow 0$.
Therefore, it follows from Theorem~\ref{theorem.asymptotics.E} below and \eqref{eq:CombinedT3k} that
\begin{align}\label{eq:t3kFinal}
\liminf_{h \rightarrow 0_{+}} h^2 {\mathcal T}_{3,k} \geq 0.
\end{align}

To establish \eqref{eq:t3kFinal} we needed the following result
\begin{theorem}
\label{theorem.asymptotics.E}
Let $\phi : \mathbb{R}^3 \to [0, 1]$ with $\text{supp }\phi \subset B(1)$ and such that $\phi = 1$ on $B(1/2)$, define $\phi_R(x) := \phi(x/R)$ for $rR> 0$, and let $\chi_R$ be the indicator function of the support of $\phi_R$. 
Define also, for $R, \Lambda > 0$, $D \geq 0$ and $0 < \alpha \leq 2/\pi$,
\begin{align}
\mathcal{E}_{R, \alpha, \Lambda}(A) := & \, \tr\left[\phi_R\left(\sqrt{\alpha^{-2}T(A) + \alpha^{-4}} - \alpha^{-2} - \frac{1}{|x|} - DR^{-2}\chi_R\right)\phi_R\right]_- \nonumber \\
&+ \Lambda\int|\nabla\otimes A|^2\,dx - I_R,
\end{align}
with $I_R$ as defined in \eqref{eq:IR}.

Define furthermore,
\begin{gather}
E(R, \alpha, \Lambda) := \inf_A\mathcal{E}_{R, \alpha, \Lambda}(A).
\end{gather}
Then,
\begin{gather}
\lim_{\substack{R, \Lambda \to \infty\\R^{11}/\Lambda \to 0}}\,\sup_{0 < \alpha \leq 2/\pi}\left|E(R, \alpha, \Lambda) - 2S_2(\alpha)\right| = 0.
\end{gather}
\end{theorem}
We postpone the proof of Theorem~\ref{theorem.asymptotics.E} to Section~\ref{sec:MagneticScottTerm}.

\subsection{End of the proof of Theorem \ref{theorem.preliminary.main.result} }
It follows from Theorem~\ref{thm:EFS_far} and Theorem~\ref{thm:EFS_semiclassics} and with ${\mathcal T}_1$, ${\mathcal T}_2$ from \eqref{eq:Ts} that
\begin{align}
{\mathcal T}_1 \geq -C, \quad {\mathcal T}_2 \geq - C h^{-2 + \delta/22}.
\end{align}
Combining this with \eqref{eq:t3kFinal} we find from \eqref{eq:SplittingInTs} that
\begin{align}
\liminf_{h \rightarrow 0_{+}} h^2\Bigg\{&
\tr\left[\sqrt{\beta^{-2}T_h(\widetilde{A}) + \beta^{-4}} - \beta^{-2} - \widetilde{\kappa}V_{z, r}^{TF}\right]_- + \frac{\lambda}{\beta^2 h^3}\int_{\mathbb{R}^3}\left|\nabla\otimes\widetilde{A}\right|^2 \nonumber \\
&\quad- \frac{2}{(2\pi h)^3}\int_{\mathbb{R}^3}\int_{\mathbb{R}^3}\left[\frac{p^2}{2} - \widetilde{\kappa}V_{z, r}^{TF}(x)\right]_{-}\,dx\,dp - \frac{2}{h^2}\sum_{k = 1}^M (z_k\widetilde{\kappa})^2 S_2\left(\beta h^{-1}\widetilde{\kappa}z_k\right) 
\Bigg\} \nonumber \\
&\geq 0.
\end{align}
This was what we wanted to prove.

\qed
\end{subsection}
\end{section}

\section{Proof of Theorem~\ref{theorem.asymptotics.E}}
\label{sec:MagneticScottTerm}
In the proof of Theorem~\ref{theorem.asymptotics.E} we shall use the following result.
\begin{lemma}
\label{lemma.inequality.removal.magnetic.field}
For $0 < \lambda \leq 1$, $0 < \alpha \leq 2/\pi$,
\begin{align}
&\sqrt{\alpha^{-2}T(A) + \alpha^{-4}} - \alpha^{-2} - \frac{1}{|x|} \nonumber \\
& \geq \left(1 - \lambda\right)\left(\sqrt{\alpha^{-2}p^2 + \alpha^{-4}} - \alpha^{-2} - \frac{1}{|x|}\right) - \lambda\alpha^{-2} - C\lambda^{-5/2}\alpha^{-1}\|B\|_2^2.
\end{align}
\end{lemma}
This lemma relies on results from \cite{BF}. In order not to deviate from the line of thought in the current section, the proof of Lemma~\ref{lemma.inequality.removal.magnetic.field} has been moved to the appendix (Appendix \ref{appendix.proof.bound.relativistic.pauli}). We now continue to prove Theorem \ref{theorem.asymptotics.E}.
\begin{proof}[Proof of Theorem~\ref{theorem.asymptotics.E}]~\\
\noindent{\bf Upper bound.}
The upper bound follows from Lemma~\ref{lemma.2S.non.magnetic}. Notice that for an upper bound, it suffices to consider the case $D=0$. In that case we can estimate
\begin{gather}
E(R, \alpha, \Lambda) - 2S_2(\alpha) \leq \mathcal{E}_{R, \alpha, \Lambda}(0) - 2S_2(\alpha) \leq \sup_{0 < \beta \leq 2/\pi}\left|\mathcal{E}_{R, \beta, \Lambda}(0) - 2S_2(\beta)\right| \to 0.
\end{gather}

\noindent{\bf Lower bound.}
We will start by reducing to the case $D=0$.
For this, we use Theorem~\ref{thm:MCDLY} to estimate (for $R\geq 1$),
\begin{align}\label{eq:UseDaubechies}
& \tr\Big[\phi_R\left(\sqrt{\alpha^{-2}T(A) + \alpha^{-4}} - \alpha^{-2} - \frac{1}{|x|} - DR^{-2}\chi_R\right)\phi_R\Big]_- \nonumber \\
 &\quad \geq -C \Big(  \sqrt{R} + D^{5/2} R^{-2} + D^4 R^{-5} + \alpha^{-1} \int B^2 \Big).
\end{align}
Notice that here we used that the constant term in \eqref{eq:CDLY} is not present for small values of $\alpha$ and that $R\geq 1$.

Now we let $0<\lambda <1$ and estimate
\begin{align}
& \tr\Big[\phi_R\left(\sqrt{\alpha^{-2}T(A) + \alpha^{-4}} - \alpha^{-2} - \frac{1}{|x|} - DR^{-2}\chi_R\right)\phi_R\Big]_- \nonumber \\
 &\quad \geq
 (1-\lambda) \tr\Big[\phi_R\left(\sqrt{\alpha^{-2}T(A) + \alpha^{-4}} - \alpha^{-2} - \frac{1}{|x|} \right)\phi_R\Big]_-  \nonumber\\
 &\quad
 +
 \lambda \tr\Big[\phi_R\left(\sqrt{\alpha^{-2}T(A) + \alpha^{-4}} - \alpha^{-2} - \frac{1}{|x|} - D\lambda^{-1}R^{-2}\chi_R\right)\phi_R\Big]_-
 \nonumber \\
 &\quad \geq
  (1-\lambda) \tr\Big[\phi_R\left(\sqrt{\alpha^{-2}T(A) + \alpha^{-4}} - \alpha^{-2} - \frac{1}{|x|} \right)\phi_R\Big]_- \nonumber \\
&\quad  - C \lambda \left( \sqrt{R} + \lambda^{-5/2} R^{-2} + \lambda^{-4} R^{-5} + \alpha^{-1} \int B^2
\right),
\end{align}
where we used \eqref{eq:UseDaubechies} to get the last estimate.
By choosing $\lambda := R^{-1}$  we see that it suffices to prove the lower bound of Theorem~\ref{theorem.asymptotics.E} in the case $D=0$.

In the case $D=0$ the lower (and upper) bound in Theorem~\ref{theorem.asymptotics.E} is given in \cite[Lemma 4.1]{EFS} in the case where $\alpha$ is bounded away from $2/\pi$.
Therefore, it suffices to prove the lower bound in the case $D=0$ and under the assumption
\begin{align}\label{eq:limitparams}
 1/\pi \leq \alpha \leq 2/\pi.
\end{align}
In the remainder of the proof this will be assumed.

We start by applying  \eqref{eq:UseDaubechies} in the case $D=0$.
In this way, if $A$ is such that $\mathcal{E}_{R, \alpha, \Lambda}(A) \leq \mathcal{E}_{R, \alpha, \Lambda}(0)$, and $R \geq 1$, (and remembering the lower bound on $\alpha$)
\begin{equation}
\mathcal{E}_{R, \alpha, \Lambda}(0) \geq \left(\Lambda - C\right)\int|\nabla\otimes A|^2\,dx - C\sqrt{R}.
\end{equation}
Note that
\begin{gather}
\mathcal{E}_{R, \alpha, \Lambda}(0) \leq \sup_{1/\pi \leq \beta \leq 2/\pi}\left|\mathcal{E}_{R, \beta, \Lambda}(0) - 2S_2(\beta)\right| + 2S_2(\alpha) \leq C,
\end{gather}
because, by Lemma~\ref{lemma.2S.non.magnetic},
$\sup_{1/\pi \leq \beta \leq 2/\pi}\left|\mathcal{E}_{R, \beta, \Lambda}(0) - 2S_2(\beta)\right|$ converges to zero as $R \to \infty$ and $S_2 \leq 1/4$. Therefore, for $R$ and $\Lambda$ large enough,
\begin{gather}
\int |\nabla\otimes A|^2\,dx \leq \frac{C\sqrt{R}}{\Lambda}.
\label{equation.a.priori.bound.magnetic.field}
\end{gather}
This provides us with an upper bound for the magnetic field energy. 
Notice also the simple estimate,
\begin{align}
I_R \leq C \sqrt{R}.
\end{align}

We now use Lemma \ref{lemma.inequality.removal.magnetic.field} to find, with $0 < \mu < 1$,
\begin{align}
\mathcal{E}_{R, \alpha, \Lambda}(A) \geq \, & \tr\left\{\phi_R\left[(1 - \lambda)\left(\sqrt{\alpha^{-2}p^2 + \alpha^{-4}} - \alpha^{-2} - \frac{1}{|x|} \right)
-\lambda\alpha^{-2} - C\alpha^{-1}\lambda^{-5/2}\|B\|_2^2\right]\phi_R\right\}_-\nonumber\\
& + \Lambda\|B\|_2^2 - I_R\nonumber\\
\geq \, & (1 - \lambda)(1 - \mu)\left\{\tr\left[\phi_R\left(\sqrt{\alpha^{-2}p^2 + \alpha^{-4}} - \alpha^{-2} - \frac{1}{|x|} \right)\phi_R\right]_- - I_R\right\}\nonumber\\
& + (1 - \lambda)\mu\tr\Bigg\{\phi_R\bigg[\sqrt{\alpha^{-2}p^2 + \alpha^{-4}} - \alpha^{-2} - \frac{1}{|x|} \nonumber\\
& \qquad\qquad\qquad\qquad\qquad - C\lambda\mu^{-1}\alpha^{-2} - C\alpha^{-1}\lambda^{-5/2}\mu^{-1}\|B\|_2^2\bigg]\phi_R\Bigg\}_-\nonumber\\
& - (\mu + \lambda - \lambda\mu)I_R.
\end{align}
Using Lemma~\ref{lemma.2S.non.magnetic} it therefore suffices to prove that the last two terms above tend to zero in the limit.
We are led to studying
\begin{align}
\mu\tr\bigg\{\phi_R\bigg[\sqrt{\alpha^{-2}p^2 + \alpha^{-4}} - \alpha^{-2} - \frac{1}{|x|} 
 -C\lambda\mu^{-1}\alpha^{-2} - C\alpha^{-1}\lambda^{-5/2}\mu^{-1}\|B\|_2^2\bigg]\phi_R\bigg\}_-.
\end{align}
To simplify expressions, we introduce the (optimal) choice
\begin{align*}
\lambda=\mu = \alpha^{2/7} \|B \|_2^{4/7}.
\end{align*}
Notice, using \eqref{equation.a.priori.bound.magnetic.field} that with this choice $\mu,\lambda \rightarrow 0$ in the limit considered, and also
\begin{align}
(\mu + \lambda - \lambda\mu)I_R \rightarrow 0.
\end{align}
With our choices of parameters and the lower bound on $\alpha$ we have, using Theorem~\ref{thm:MCDLY},
\begin{align}
&\mu\tr\bigg\{\phi_R\bigg[\sqrt{\alpha^{-2}p^2 + \alpha^{-4}} - \alpha^{-2} - \frac{1}{|x|} 
 -C\lambda\mu^{-1}\alpha^{-2} - C\alpha^{-1}\lambda^{-5/2}\mu^{-1}\|B\|_2^2\bigg]\phi_R\bigg\}_- \nonumber \\
 &= \mu\tr\bigg\{\phi_R\bigg[\sqrt{\alpha^{-2}p^2 + \alpha^{-4}} - \alpha^{-2} - \frac{1}{|x|} 
 -C\bigg]\phi_R\bigg\}_-\nonumber \\
 &\geq -C \mu ( R^3 + \| B \|_2^2) \nonumber \\
 &\rightarrow 0,
\end{align}
since $R^{11} \Lambda^{-1} \rightarrow 0$.
Combing the estimates above, we therefore have
\begin{gather}
\lim_{\substack{R, \Lambda \to \infty\\R^{11}/\Lambda \to 0}}\sup_{\alpha \in [1/\pi, 2/\pi]}\left|E(R, \alpha, \Lambda) - 2S_2(\alpha)\right| = 0.
\end{gather}
\end{proof}

\section{Proof of Theorem~\ref{thm:MCDLY}}
\label{sec:PfofDaubechies}

\begin{proof}
\noindent
{\bf Case 1. $\nu \alpha \in [  \frac{1}{64 M}, \frac{2}{\pi}]$.}\\
Define $r:= (\pi+1)$.
Define $\theta(t) = \begin{cases} 1, & t \leq \alpha \\
\cos(\frac{t-\alpha}{2\alpha}), & \alpha < t \leq r \alpha \\
0, & t > r \alpha
\end{cases}$, and define 
\begin{align}
\theta_j(x) &= \theta(|x-R_j|), \text{ for } j \in \{ 1,\ldots, M\}, \nonumber \\
\theta_{M+1}(x) & = \sqrt{1 - \sum_{j=1}^M \theta_j^2(x)}.
\end{align}
Then, using the separation condition \eqref{eq:SepCond} we have
\begin{align}
|\nabla \theta_j| \leq \frac{1}{2\alpha} \chi_{\Omega_j},
\end{align}
with $\chi$ denoting the indicator function and
\begin{align}
\Omega_j = \begin{cases}
B(R_j, r \alpha)\setminus B(R_j, \alpha), & j \in \{ 1, \ldots, M\}, \\
\cup_{k=1}^M \Omega_k,  & j= M+1.
\end{cases}
\end{align}
Therefore, by the IMS-localization formula,
\begin{align*}
\alpha^{-2} [\sigma\cdot p_{A}]^2 &\geq
\alpha^{-2} \sum_{j=1}^{M+1}\theta_j \left(
[\sigma\cdot p_{A}]^2 - \alpha^{-2} \chi_{\Omega_j}/4\right) \theta_j.
\end{align*}
So, using operator-monotonicity of the square root and the pull-out formula (Prop.~\ref{proposition.pull.out}),
\begin{align}
\sqrt{ \alpha^{-2} [\sigma\cdot p_{A}]^2 + \alpha^{-4}}
&\geq
\sum_{j=1}^{M} \theta_j
\alpha^{-1} |\sigma\cdot p_{A}|  \theta_j  
%\nonumber \\
%&\quad 
+
\theta_{M+1}
\sqrt{
\alpha^{-2} [\sigma\cdot p_{A}]^2 + \alpha^{-4} (1-\frac{1}{4}\chi_{\Omega_{M+1}})} \,\theta_{M+1}.
\end{align}
Also,
\begin{align}
W &= \sum_{j=1}^M \theta_j W \theta_j + \theta_{M+1}  W \theta_{M+1}  
%\\
%&
\geq
\sum_{j=1}^M  \theta_j^2 \left( - \frac{\nu}{|x-R_j|} - U\right)
- \theta_{M+1}  U \theta_{M+1}.
\end{align}
We can therefore estimate
\begin{align}\label{eq:Splitting}
&\tr\left( \sqrt{\alpha^{-2}[\sigma\cdot p_A]^2 + \alpha^{-4}} - \alpha^{-2}+ W \right)_{-} \nonumber \\
&\quad\geq
\alpha^{-1} \sum_{j=1}^M \tr\left(\theta_j( |\sigma\cdot p_{A}|  - \alpha^{-1} - [ \frac{\alpha \nu}{|x-R_j|} + \alpha U  ]) \theta_j \right)_{-} \nonumber \\
&\quad\quad +
\tr\left(
\theta_{M+1}\big(
\sqrt{
\alpha^{-2} [\sigma\cdot p_{A}]^2 + \alpha^{-4} (1-\frac{1}{4}\chi_{\Omega_{M+1}})} 
- \alpha^{-2} - U \big) \theta_{M+1}
\right)_{-} .
\end{align}
For $j\in \{1,\ldots, M\}$,
we have by the Pauli-Hardy-Lieb-Thirring inequality \eqref{inequality.pauli.hardy.lieb.thirring}, and using that $\nu\alpha \leq \frac{2}{\pi}$ and $\theta_j$ localizes to a ball of radius proportional to $\alpha$
\begin{align}
\alpha^{-1}  \tr\left(\theta_j( |\sigma\cdot p_{A}|  - \alpha^{-1} - [ \frac{\alpha\nu}{|x-R_j|} + \alpha U  ]) \theta_j \right)_{-} 
\geq - C \alpha^{-1} \Big(\alpha^{-1}  + \alpha^4 \int U^4 +  \int B^2\Big).
\end{align}
Therefore, the contribution from the terms with $j\in \{1,\ldots, M\}$ are in agreement with \eqref{eq:CDLY}.

For the $\theta_{M+1}$-term in \eqref{eq:Splitting} we need to prove spatial decay of the localization error. For this we use a dyadic-type partition of unity on ${\mathbb R}^3$ (similar to \cite[Sect.3.1]{EFS}). 
The partition of unity is given as
\begin{align}
\sum_{j=0}^{\infty} \phi_{j,\alpha}^2 =1,
\end{align}
where, for some parameter $T\geq 1$,
\begin{align}
\supp \phi_{0,\alpha} &\subset \{ d_R(x) \leq T \alpha\}, \nonumber \\
\supp \phi_{j,\alpha} &\subset \{ T 2^{j-2} \alpha \leq d_R(x) \leq T 2^{j} \alpha \}, \quad j\geq 1,\nonumber \\
|\nabla \phi_{j, \alpha}| & \leq C 2^{-j} \alpha^{-1} T^{-1}.
\end{align}
Therefore, by the standard localization argument
\begin{align}
 [\sigma\cdot p_{A}]^2 \geq \sum_{j=0}^{\infty} \phi_{j,\alpha} \left( [\sigma\cdot p_{A}]^2 - C T^{-2} 2^{-2j} \alpha^{-2}\right)  \phi_{j,\alpha} .
\end{align}
We will choose $T \geq 4 r$. Then, $\phi_{j,\alpha} \chi_{\Omega_{M+1}}=0$, for all $j\geq1$.

For $j= M+1$, we have, where we have chosen and fixed $T$ sufficiently large,
\begin{align}
\alpha^{-2} [\sigma\cdot p_{A}]^2 + \alpha^{-4} (1-\frac{1}{4}\chi_{\Omega_{M+1}})
&\geq
\phi_{0,\alpha} \alpha^{-2}  [\sigma\cdot p_{A}]^2 \phi_{0,\alpha} \nonumber \\
&\quad +
\sum_{j=1}^{\infty} \phi_{j,\alpha} \left(
\alpha^{-2} [\sigma\cdot p_{A}]^2 + \alpha^{-4} (1-  2^{-2j} )
\right)
 \phi_{j,\alpha} 
\end{align}
Therefore, by the pull-out formula (Prop.~\ref{proposition.pull.out}),
\begin{align}
&\theta_{M+1}\left(
\sqrt{
\alpha^{-2} [\sigma\cdot p_{A}]^2 + \alpha^{-4} (1-\frac{1}{4}\chi_{\Omega_{M+1}})} - \alpha^{-2} \right)\theta_{M+1} \nonumber \\
&\geq
\alpha^{-1} \theta_{M+1} \phi_{0,\alpha} \left(  |\sigma\cdot p_A|  - \alpha^{-1}\right) \phi_{0,\alpha} \theta_{M+1} 
\nonumber \\
&\quad +
\sum_{j=1}^{\infty} 
\theta_{M+1} \phi_{j,\alpha} \left(
\sqrt{
\alpha^{-2} [\sigma\cdot p_{A}]^2 + \alpha^{-4} (1-  2^{-2j} )} - \alpha^{-2}\right) \phi_{j,\alpha}\theta_{M+1} \nonumber \\
&\geq
\alpha^{-1} \theta_{M+1}\phi_{0,\alpha}  \left(  |\sigma\cdot p_A|  - \alpha^{-1}\right) \phi_{0,\alpha} \theta_{M+1} 
\nonumber \\
&\quad +
\sum_{j=1}^{\infty} 
\theta_{M+1} \phi_{j,\alpha} \left(
\sqrt{
\alpha^{-2} [\sigma\cdot p_{A}]^2 + \alpha^{-4} } - \alpha^{-2}(1+ C 2^{-2j})\right) \phi_{j,\alpha}\theta_{M+1} ,
\end{align}
where we used that $\xi \mapsto \sqrt{a^2 + \xi^2} -\xi$ in non-increasing on $[0,\infty)$ for all real $a$, and we estimated
$(1- 2^{-2j})^{-1} - 1 \leq C 2^{-2j}$.

At this point we apply the inequality of Theorem~\ref{theorem.lieb.thirring.efs} to get that
\begin{align}
&\tr\left(
\theta_{M+1}\left(
\sqrt{
\alpha^{-2} [\sigma\cdot p_{A}]^2 + \alpha^{-4} (1-\frac{1}{4}\chi_{\Omega_j})} - \alpha^{-2} - U \right)\theta_{M+1}
\right)_{-} 
\nonumber \\
&\geq
\tr\left(
\theta_{M+1}\phi_{0,\alpha}  \left(  \alpha^{-1}|\sigma\cdot p_A|  - \alpha^{-2} - U\right) \phi_{0,\alpha} 
\theta_{M+1}
\right)_{-} 
\nonumber \\
&+
\sum_{j=1}^{\infty}
\tr\left(
\theta_{M+1}\phi_{j,\alpha} \left(
\sqrt{
\alpha^{-2} [\sigma\cdot p_{A}]^2 + \alpha^{-4} } - \alpha^{-2} - C 2^{-2j} \alpha^{-2} 1_{\supp \phi_{j,\alpha}}-U \right) \phi_{j,\alpha}
\theta_{M+1}
\right)_{-} 
\nonumber \\
&\geq
-C \left\{ \alpha^{-2}+ 
\int  U^{5/2} + \alpha^3 \int U^4 + \alpha^{-1} \int B^2\right\} .
\end{align}
Here we used, in particular, that for $p> 3/2$, 
\begin{align*}
\sum_{j=0}^{\infty} \int |2^{-2j} \alpha^{-2} 1_{\supp \phi_{j,\alpha}}|^p \leq C \alpha^{-2p+3}.
\end{align*}
Clearly, this finishes the proof of Theorem~\ref{thm:MCDLY} in the case $\nu \alpha \in [ \frac{1}{64 M }, \frac{2}{\pi}]$.

\noindent
{\bf Case 2. $\nu \alpha < \frac{1}{64 M}$.}\\
We end by treating the easier case of small values of $\nu \alpha$.
In this case we will never be near criticality and can use less delicate estimates.
We write
\begin{align}
&\tr\left( \sqrt{\alpha^{-2}[\sigma\cdot p_A]^2 + \alpha^{-4}} - \alpha^{-2}+ W \right)_{-} \nonumber \\
&\geq
\frac{1}{2} \tr\left( \sqrt{\alpha^{-2}[\sigma\cdot p_A]^2 + \alpha^{-4}} - \alpha^{-2} - 2 U \right)_{-} \nonumber \\
&\quad +
\frac{1}{2M} \sum_{j=1}^M 
\tr\left( \sqrt{\alpha^{-2}[\sigma\cdot p_A]^2 + \alpha^{-4}} - \alpha^{-2}- \frac{2 M \nu}{|x-R_j|}
1_{\{|x-R_j| \leq \alpha\}} 
\right)_{-}
\end{align}
On the first term we can apply the Lieb-Thirring-type inequality of Theorem~\ref{theorem.lieb.thirring.efs}.
\begin{align}\label{eq:Integrals}
 \tr\left( \sqrt{\alpha^{-2}[\sigma\cdot p_A]^2 + \alpha^{-4}} - \alpha^{-2} -2 U \right)_{-} 
 \geq
 - C \left(
 \int U^{5/2} + \alpha^3  \int U^4 + \alpha^{-1} \int B^2
 \right).
\end{align}
This is in agreement with the integral terms in \eqref{eq:CDLY}.

To estimate the terms near the singularities we will use the BKS inequality.
All $M$ terms will be treated in the same manner, so for notational simplicity we only treat one term and assume that $R_j=0$.

This yields, 
\begin{align}
&\tr\left( \sqrt{\alpha^{-2}[\sigma\cdot p_A]^2 + \alpha^{-4}} - \alpha^{-2}- \frac{2 M \nu}{|x|} 1_{\{|x| \leq \alpha\}} 
\right)_{-} \nonumber \\
&\geq - \alpha^{-1} \tr \left\{-\left(
[\sigma\cdot p_A]^2 + \alpha^{-2}
- [\alpha^{-1}+\frac{2 M \alpha \nu}{|x|} 1_{\{|x| \leq \alpha\}}  ]^2
\right)_{-} 
\right\}^{1/2} \nonumber \\
&\geq
-\alpha^{-1}  \tr \left\{-\left(
p_A^2 - \frac{1}{4|x|^2}-|B|
+ \left[ \frac{1}{4|x|^2}
-  \frac{ 4M \nu}{|x|}1_{\{|x| \leq \alpha\}} 
-   \frac{(2M\alpha \nu)^2}{|x|^2} \right]
\right)_{-} 
\right\}^{1/2} .
\end{align}
By the interval for $\alpha \nu$ the term in square brackets $[\cdots]$ is non-negative and can be dropped.
Therefore, we get from the magnetic version of the Hardy-Lieb-Thirring inequality for the $\frac{1}{2}$-moments of the eigenvalues for the Schr\"{o}dinger operator \cite[Theorems 1.1 and 4.2]{RF}
\begin{align}\label{eq:conclusion}
\tr\left( \sqrt{\alpha^{-2}[\sigma\cdot p_A]^2 + \alpha^{-4}} - \alpha^{-2}- \frac{2 M \nu}{|x|} 1_{\{|x| \leq \alpha\}} 
\right)_{-} \geq
-\alpha^{-1} C
\int B^2 .
\end{align}
Clearly, this is also in agreement with \eqref{eq:CDLY}.

Notice that neither in \eqref{eq:Integrals} nor \eqref{eq:conclusion} did we get a contribution to the constant term in \eqref{eq:CDLY}.

This finishes the proof of Theorem~\ref{thm:MCDLY}.
\end{proof}

\appendix
\section{Proof of Lemma \ref{lemma.inequality.removal.magnetic.field}}
\label{appendix.proof.bound.relativistic.pauli}
Let $m \geq 0$, $0 \leq c \leq 2/\pi$, and $0 < \lambda \leq 1$. We shall prove
\begin{gather}
\sqrt{[\sigma\cdot p_A]^2 + m^2} - m - \frac{c}{|x|} \geq (1 - \lambda)\left(\sqrt{p^2 + m^2} - m - \frac{c}{|x|}\right) - \lambda m - C\lambda^{-5/2}\|B\|_2^2.
\end{gather}
Estimate I, Theorem 2.3 in \cite{BF}, says in particular that, with the choice of adequate parameters, for any $\varepsilon > 0$ and normalized $\psi \in Q([\sigma\cdot p_A]^2)\cap Q\left( (p - A)^2\right)$, $Q$ denoting the quadratic form domain of an operator,
\begin{gather}
\Big(\psi, \sqrt{
[\sigma\cdot p_A]^2 + m^2}\,\psi\Big) \geq \Big(\psi, \sqrt{p_A^2 + m^2}\,\psi\Big) - \varepsilon\left(\psi, |p_A|^{3/4}\psi\right)^{4/3} - F(\varepsilon)\|B\|_2^2,
\end{gather}
for some function $F > 0$. Furthermore, Estimate II, Theorem 2.6 in \cite{BF}, says that
\begin{gather}
\left|\left(\psi, \left(|p_A|^{3/4} - |p|^{3/4}\right)\psi\right)\right| \leq \varepsilon\left(\psi, |p|^{3/4}\psi\right) + G(\varepsilon)\|B\|_2^{3/2},
\end{gather}
and
\begin{gather}
\Big(\psi, \sqrt{p_A^2 + m^2}\,\psi\Big) \geq \left(\psi, \sqrt{p^2 + m^2}\,\psi\right) - \varepsilon\left(\psi, |p|^{3/4}\psi\right)^{4/3} - H(\varepsilon)\|B\|_2^{3/2},
\end{gather}
for some functions $G, H > 0$. Therefore,
\begin{gather}
\left(\psi, |p_A|^{3/4}\psi\right)^{4/3} \leq 2^{1/3}\left(2^{1/3}\varepsilon^{4/3} + 1\right)\left(\psi, |p|^{3/4}\psi\right)^{4/3} + 2^{2/3}G(\varepsilon)^{4/3}\|B\|_2^2,
\end{gather}
and so, if $J(\varepsilon) \equiv 2^{2/3}G(\varepsilon)^{4/3}\varepsilon + F(\varepsilon) + H(\varepsilon)$, and $\varepsilon$ is chosen smaller than 1,
\begin{align}
&\left(\psi, \left(\sqrt{[\sigma\cdot p_A]^2 + m^2} - m - \frac{c}{|x|}\right)\psi\right)\nonumber \\
& \geq \,  \left(\psi, \left(\sqrt{p^2 + m^2} - m - \frac{c}{|x|}\right)\psi\right)
- \left(2^{2/3}\varepsilon^{7/3} + 2^{1/3}\varepsilon + \varepsilon\right)\left(\psi, |p|^{3/4}\psi\right)^{4/3} - J(\varepsilon)\|B\|_2^2\nonumber\\
&\geq \,  \left(\psi, \left(\sqrt{p^2 + m^2} - m - \frac{c}{|x|}\right)\psi\right) - 5\varepsilon\left(\psi, |p|^{3/4}\psi\right)^{4/3} - J(\varepsilon)\|B\|_2^2\nonumber\\
&\geq \,  \left(1 - \lambda\right)\left(\psi, \left(\sqrt{p^2 + m^2} - m - \frac{c}{|x|}\right)\psi\right)\nonumber\\
&\quad + \lambda\left[\left(\psi, \left(|p| - \frac{c}{|x|}\right)\psi\right) - \frac{5\varepsilon}{\lambda}\left(\psi, |p|^{3/4}\psi\right)^{4/3}\right] - m\lambda - J(\varepsilon)\|B\|_2^2.
\label{estimate.P.A.p.almost.finished}
\end{align}
If we now set $\varepsilon = \lambda d$, with $d > 0$ a small enough fixed constant, then the term in brackets in \eqref{estimate.P.A.p.almost.finished} is non-negative, by Equation (1.25) in \cite{BF}. (The result that the form $|p| - c/|x|$ can control a power of $|p|$ smaller than 1 appeared first in \cite[Theorem 2.3]{SSS}; see also \cite[Theorem 1.2]{RF}.) We conclude in this way that
\begin{gather}
\sqrt{[\sigma\cdot p_A]^2 + m^2} - m - \frac{c}{|x|} \geq \left(1 - \lambda\right)\left(\sqrt{p^2 + m^2} - m - \frac{c}{|x|}\right) - m\lambda - J(d\lambda)\|B\|_2^2.
\end{gather}
Finally, it is easy to show that $J(d\lambda) \leq C\lambda^{-5/2}$ by using the explicit expressions for $F$, $G$ and $H$ appearing in \cite[Appendix A.1]{BF}. This finishes the proof of Lemma \ref{lemma.inequality.removal.magnetic.field}.
\qed

\paragraph{Acknowledgement}
SF was partially supported by the Sapere Aude grant DFF–4181- 00221 from the Independent Research Fund Denmark.

\end{document}